\documentclass[11pt,pdftex,a4paper]{article}
\pdfoutput=1
\usepackage{verbatim}
\usepackage{amssymb}
\usepackage{lipsum}
\usepackage{multicol}
\usepackage{amsmath}
\usepackage{framed}
\usepackage{amsthm}
\usepackage{tikz}
\usepackage{color}
\usepackage{fullpage}
\usepackage{cite}
\newif\ifcode
\codefalse
\codetrue
\usepackage{macros}
\ifcode
\usepackage[ruled]{algorithm}
\usepackage{algpseudocode}
\usepackage{ioa_code}

\newtheorem{theorem}{Theorem}
\newtheorem{claim}[theorem]{Claim}
\newtheorem{proposition}{Proposition}

\newtheorem{corollary}[theorem]{Corollary}
\newtheorem{definition}{Definition}
\newtheorem{lemma}[theorem]{Lemma}

\newcounter{linenumber}

\def\Nat{\ensuremath{\mathbb{N}}}

\newcommand{\remove}[1]{}

\newcommand{\Wset}{\textit{Wset}}
\newcommand{\Rset}{\textit{Rset}}

\newcommand{\txns}{\textit{txns}}

\newcommand{\Read}{\textit{read}}
\newcommand{\Write}{\textit{write}}
\newcommand{\TryC}{\textit{tryC}}
\newcommand{\TryA}{\textit{tryA}}

\newcommand{\ignore}[1]{}

\begin{document}

\bibliographystyle{plain}

\title{Safety of Deferred Update \\ in Transactional Memory
%\title{What is Safe in Transactional Memory?}
}

\author{Hagit Attiya$^1$
~~~
Sandeep Hans$^1$
~~~
Petr Kuznetsov$^{2,3}$
~~~
Srivatsan Ravi$^2$\\
\\
$^1$\normalsize Department of Computer Science, Technion\\
$^2$\normalsize Telekom Innovation Laboratories, TU Berlin\\
$^3$\normalsize T\'el\'ecom ParisTech
}
%\\
%\\
%\normalsize TU Berlin/Deutsche Telekom Laboratories
%\and
%Srivatsan Ravi\\
%\\
%\normalsize TU Berlin/Deutsche Telekom Laboratories
%%\thanks{TU Berlin/ Deutsche Telekom Laboratories
%%FG INET, Sekr. TEL 16, Ernst-Reuter-Platz 7, 10587 Berlin; Email: srivatsan.ravi@net.t-labs.tu-berlin.de}

%\institute{Deutsche Telekom Laboratories/TU Berlin}

%%%%%%%%%%%%%%%%%%%%%%%%%%%%%%%%%%%%%%%%%%%%%%%%%%%%%%%%%%%%%%%%%%%%%%%%%%%%%%%%
\date{}
\maketitle

\begin{abstract}
Transactional memory allows the user to declare sequences of
instructions as speculative \emph{transactions} that can either \emph{commit}
or \emph{abort}.
If a transaction commits, it appears to be executed sequentially,
so that the committed transactions constitute a correct
sequential execution.
If a transaction aborts, none of its instructions can affect
other transactions.

The popular criterion of \emph{opacity} requires that
the views of aborted transactions must also be consistent with the global
sequential order constituted by committed ones.
This is believed to be important, since inconsistencies observed by
an aborted transaction may cause a fatal irrecoverable error or waste
of the system in an infinite loop.
Intuitively, an opaque implementation must ensure that no intermediate view a
transaction obtains before it commits or aborts can be
affected by a transaction that has not started committing yet, so
called \emph{deferred-update} semantics.

In this paper, we intend to grasp this intuition formally.
We propose a variant of opacity that explicitly requires the
sequential order to respect the deferred-update semantics.
%We show that our criterion is a safety property, i.e., it is prefix-
%and limit-closed.
Unlike opacity, our property also ensures that a serialization of a
history implies serializations of its prefixes.
Finally, we show that our property is equivalent to opacity if we assume that
no two transactions commit identical values on the same variable, and present a
counter-example for scenarios when the ``unique-write'' assumption
does not hold.
\end{abstract}

%\begin{center}
%Regular paper, eligible for the best student paper award (Srivatsan Ravi is a full-time student).
%\end{center}
%
%\thispagestyle{empty}
%\clearpage
\pagenumbering{arabic}

%%%%%%%%%%%%%%%%%%%%%%%%%%%%%%%%%%%%
%%%%%%%%%%%%%%%%%%%%%%%%%%%%%%%%%%%%
%\section{Introduction}
%
\section{Introduction}
\label{sec:intro}
Resolving conflicts in an efficient and consistent manner is the most
challenging task in concurrent software design.
Transactional memory (TM)~\cite{HM93,ST95}
addresses this challenge by offering an interface in which  sequences
of shared-memory instructions can be declared as speculative \emph{transactions}.
The underlying idea, borrowed from databases, is to treat each
transaction as an atomic event: a  transaction may either \emph{commit} in
which case it appears as executed sequentially, or \emph{abort} in
which case none of its update instructions affect other transactions.
The user can therefore design software having only sequential
semantics in mind and let the memory take care of
conflicts resulting from potentially concurrent executions.

In databases, a correct implementation of concurrency control
should guarantee that committed transactions constitute a
serial (or sequential) execution~\cite{Had88}.
On the other hand, uncommitted transactions can be aborted without
invalidating the correctness of committed ones.
(In the literature on databases, the latter feature is called
\emph{recoverability}.)
%, and it can be implemented by
%rolling back all the updates performed by the aborted transaction.

%\petrC{A better link here?}\hagitC{Yes}
In the TM context,
%traditional database conditions
%may not be sufficient for ensuring system correctness,
%because
intermediate states witnessed by an incomplete transaction
may affect the application through the outcome of its read operations.
%\hagitC{DB also has a condition that requires live transactions
%to see consistent views,
%though not exactly the same as opacity.}
%\petrC{We do not say DB does not have it, what we say is that
%recoverability is not enough for TM. Is it confusing?}
If the intermediate state is not consistent with any sequential
execution, the application may experience a fatal irrecoverable error or sink
in an infinite loop.
The correctness criterion of
\emph{opacity}~\cite{GK08-opacity,tm-book}
addresses this issue by
requiring the states observed by uncommitted transactions to
be consistent with a global
serial execution constituted by committed ones
(a \emph{serialization}).

An opaque TM implementation  must, intuitively, ensure that no
transaction can read from a transaction that has not started
committing yet.
This is usually referred to as the \emph{deferred-update}
semantics, and it was in fact explicitly required in some representations of
opacity~\cite{GHS08-permissiveness}.
The motivation of this paper is to capture this intuition formally.

We present a new correctness criterion called \emph{du-opacity}.
Informally, a du-opaque (possibly, non-serial) execution must be
indistinguishable from a totally-ordered execution, with respect to which
%\hagitC{in which???}
%\petrC{I think it should be ``with respect to which''}
no transaction reads from a transaction that has not started committing.

We further check if our correctness criterion is a \emph{safety
property}, as defined by Owicki and Lamport~\cite{OL82},
Alpern and Schneider~\cite{AS85} and refined by Lynch~\cite{Lyn96}.
We show that du-opacity is \emph{prefix-closed}:
every prefix of a du-opaque history is also du-opaque.
We also show that du-opacity is,
under certain restrictions, \emph{limit-closed}.
More precisely, assuming that, in an infinite execution, every transaction
completes each of the operations it invoked (but possibly neither commits nor aborts),  the infinite limit of any sequence
of ever extending du-opaque histories is also du-opaque.
To prove that such an implementation is du-opaque, it is thus
sufficient to prove that all its \emph{finite} histories are du-opaque.
To the best of our knowledge, this paper contains the first non-trivial proof of
limit-closure for a TM correctness property.
We further show that any du-opaque serialization of a history implies
a serialization of any of its prefixes that maintains the original read-from
relations, which is instrumental in the comparison of du-opacity with opacity.

Opacity, as defined in~\cite{tm-book},
reduces correctness of an infinite history to correctness of all
its prefixes, and thus is limit-closed by definition.
In fact, we show that extending opacity to infinite histories in a
non-trivial way (i.e., requiring that even infinite histories should have proper
serializations), does not result in a limit-closed property.
We observe that opacity does not preclude
scenarios in which a transaction
reads from a future transaction (cf. example in
Figure~\ref{fig:lin-example}), and, thus, our
criterion is strictly stronger than opacity.
Surprisingly, this is true even if we assume that
all transactional operations are atomic, which somewhat attenuates
earlier attempts to forcefully introduce the deferred-update
%semantics previously used in the definition of
in the definition of opacity for atomic operations~\cite{GHS08-permissiveness}.
However, we show that opacity and du-opacity are equivalent
if we assume that no two transactions try to
commit identical values on the same data item.

We believe that these results improve our understanding of the very
notion of correctness in transactional memory. Our correctness
criterion
explicitly declares that a transaction is not allowed to
read from a transaction that has not started
committing yet, and we conjecture
that it is simpler to verify. We present the first non-trivial proof for both limit-
and prefix-closure of TM histories, which is quite
interesting in its own right, for it enables reasoning about possible
serializations of an infinite TM history based on serializations of
its prefixes.

The paper is organized as follows.
In Section~\ref{sec:modeltm}, we introduce our basic model definitions
and recall the notion of safety~\cite{OL82,AS85,Lyn96}.
In Section~\ref{sec:kr}, we define our criterion of du-opacity and
%in Section~\ref{sec:saf}, we
show that it is prefix-closed and under certain restrictions, a limit-closed property.
In Section~\ref{sec:gkopacity}, we prove that du-opacity
is a proper subset of the original notion of opacity~\cite{tm-book}, and that it
coincides with du-opacity under the ``unique-writes'' condition.
%
%\input{intro}
%
%%%%%%%%%%%%%%%%%%%%%%%%%%%%%%%%%%%%
\section{Model}
\label{sec:modeltm}
A \emph{transactional memory} (in short, \emph{TM})
supports atomic \emph{transactions} for
reading and writing on a set of \emph{transactional} objects
(in short, \emph{t-objects}).
A transaction is a sequence of accesses (reads or writes) to t-objects;
each transaction $T_k$ has a unique identifier $k$.

A transaction $T_k$ may contain the following \emph{t-operations},
each being a matching pair of \emph{invocation} and \emph{response} events:
\begin{enumerate}
\item \textit{read}$_k(X)$ returns a value in some domain $V$
or a special value $A_k\notin V$ (\emph{abort});
\item \textit{write}$_k(X,v)$, for a value $v \in V$,
returns \textit{ok}$_k$ or $A_k$;
\item \textit{tryC}$_k$ returns $C_k\notin V$ (\emph{commit}) or $A_k$; and
\item \textit{tryA}$_k$ returns $A_k$.
\end{enumerate}

The \emph{read set} (resp., the \emph{write set}) of a transaction $T_k$,
denoted $\Rset(T_k)$, is the set of t-objects that $T_k$ reads in $H$;
the \emph{write set} of $T_k$, denoted $\Wset(T_k)$,
is the set of t-objects $T_k$ writes to in $H$.
%The \emph{data set} of $T_k$ is $\Dset(T_k)=\Rset(T_k)\cup\Wset(T_k)$.

%\vspace{1mm}\noindent\textit{Histories.}
We consider an asynchronous shared-memory system in which
processes communicate via transactions.
A TM \emph{implementation} provides processes with algorithms
for implementing $\Read_k$, $\Write_k$, $\TryC_k()$ and $\TryA_k()$
of a transaction $T_k$.

%\vspace{1mm}\noindent\textit{Histories.}

A \emph{history} of a TM implementation is a (possibly infinite) sequence of
invocation and response \emph{events} of t-operations.
%, let $<_H$
%denote the total order on the events in $H$.

For every transaction identifier $k$,
$H|k$ denotes the subsequence of $H$ restricted to events of
transaction $T_k$.
%Let $\ms{inv}(op_m)$ denote the invocation of some t-operation $op_m$.
If $H|k$ is non-empty,
%and $H|k \neq \ms{inv}(op_m) \cdot A_k \vee \TryC_k() \cdot C_k$,
we say that $T_k$ \emph{participates} in $H$,
and let $\txns(H)$ denote the set of transactions that participate in $H$.
In an infinite history $H$, 
we assume that each $T_k\in \txns(H)$, $H|k$ is finite;
i.e., transactions do not issue an infinite number of t-operations.

Two histories $H$ and $H'$ are \emph{equivalent} if $\txns(H) = \txns(H')$
and for every transaction $T_k \in \txns(H)$, $H|k=H'|k$.

A history $H$ is \emph{sequential} if every invocation of
a t-operation is either the last event in $H$ or
is immediately followed by a matching response.

A history is \emph{well-formed} if for all $T_k$, $H|k$ is
sequential and has no events after $A_k$ or $C_k$.
We assume that all histories are well-formed, i.e.,
the client of the transactional memory never invokes a t-operation
before receiving a response from the previous one and
does not invoke any t-operation $op_k$ after receiving $C_k$ or $A_k$.
We also assume, for simplicity, that the client invokes a $\Read_k(X)$
%(resp. $\Write_k(X)$)
at most once within a transaction $T_k$.
%if it has previously invoked a $\Read_k(X)$ (resp. $\Write_k(X)$).
This assumption incurs no loss of generality,
since a repeated read can be assigned to return a previously returned
value without affecting the history's correctness.

A transaction $T_k\in \txns(H)$ is \emph{complete in $H$} if
$H|k$ ends with a response event.
The history $H$ is \emph{complete} if all transactions in $\txns(H)$
are complete in $H$.

A transaction $T_k\in \txns(H)$ is \emph{t-complete} if $H|k$
ends with $A_k$ or $C_k$; otherwise, $T_k$ is \emph{t-incomplete}.
$T_k$ is \emph{committed} (resp., \emph{aborted}) in $H$
if the last event of $T_k$ is $C_k$ (resp., $A_k$).
The history $H$ is \emph{t-complete} if all transactions in
$\txns(H)$ are t-complete.

For t-operations $op_k, op_j$, we say that $op_k$ \emph{precedes}
$op_j$ in the \emph{real-time order} of $H$,
denoted $op_k\prec_H^{RT} op_m$,
if the response of $op_k$ precedes the invocation of $op_j$.

Similarly, for transactions $T_k,T_m\in \txns(H)$, we say that $T_k$ \emph{precedes}
$T_m$ in the \emph{real-time order} of $H$, denoted $T_k\prec_H^{RT} T_m$,
if $T_k$ is t-complete in $H$ and
the last event of $T_k$ precedes the first event of $T_m$ in $H$.
If neither $T_k\prec_H^{RT} T_m$ nor $T_m\prec_H^{RT} T_k$,
then $T_k$ and $T_m$ \emph{overlap} in $H$.
A history $H$ is \emph{t-sequential} if there are no overlapping
transactions in $H$.

For simplicity of presentation, we assume that each history $H$
begins with an ``imaginary'' transaction $T_0$ that writes initial
values to all t-objects and commits before any other transaction
begins in $H$.

Let $H$ be a %t-complete
t-sequential history.
For every operation $\Read_k(X)$ in $H$,
we define the \emph{latest written value} of $X$ as follows:
\begin{enumerate}
\item If $T_k$ contains a $\Write_k(X,v)$ preceding $\Read_k(X)$,
then the latest written value of $X$ is the value of the latest such write to $X$.
\item Otherwise, if $H$ contains a $\Write_m(X,v)$,
$T_m$ precedes $T_k$, and $T_m$ commits in $H$,
then the latest written value of $X$ is the value
of the latest such write to $X$ in $H$.
(This write is well-defined since $H$ starts with $T_0$ writing to
all t-objects.)
%such that $T_m \prec_H T_k$, and there does not exist
%any committed transaction $T_i$ that performs $\Write_i(X,v');v' \neq v$ in $H$ such that $T_m \prec_H T_i \prec_H T_k$, then the latest written value of $X$ is the value of the earliest such write to $X$ in $H$.
%$T_m$ precedes $T_k$, and $T_m$ commits in $H$, then
%the latest written value of $X$ is the value of the latest such write to $X$ in $H$.
%\item Otherwise, the latest written value of $X$ is the initial value of $X$.
\end{enumerate}
We say that $\Read_k(X)$ is \emph{legal} in a t-sequential history $H$ if it returns the
latest written value of $X$, and $H$ is \emph{legal}
if every $\Read_k(X)$ in $H$ that does not return $A_k$ is legal in $H$.
%returns the latest written value of $X$.

\ignore{
\paragraph{Read-from relation:}
Let $S$ be a legal t-complete t-sequential history.
We say that an operation $\Read_k(X)$ that returns $v$ in $S$
\emph{reads from a transaction}  $T_m$ if $T_m$ is the
earliest committed transaction in $S$ such that (1) $T_m$ contains
$\Write_m(X,v)$, and (2) there does not exist any committed
transaction $T_i$ in $H$ such that $\TryC_m() \prec_H^{RT} \TryC_i() \prec_H^{RT}
\Read_k(X)$ that performs $\Write_i(X,v')$; $v' \neq v$.
Note that, since $S$ is legal, such a transaction $T_m$ exists for
every read operation ($T_m$ can be the initializing transaction $T_0$.

Let $H$ be any history, $S$ be a legal t-complete t-sequential
history such that $\txns(S)\subseteq\txns(H)$, and $T_k$ be a
transaction in $\txns(S)$.
We say that the read-from relation in $S$ \emph{respects the
deferred-update order with respect to $T_k$} if
for each $\Read_k(X)$ in $S$ and $T_m\in \txns(S)$, such that
$\Read_k(X)$ reads from $T_m$, we have $T_k\nprec_H^{DU}T_m$.
Intuitively, this means that no read operation in $T_m$ reads from an
incomplete or aborted transaction

\begin{lemma}
If no  two transactions write the same value in $H$ and the read-from
relation of $S$ respects , then
$\prec$
\end{lemma}

%\section{The two faces of TM safety}
%\label{sec:safetytm}
%
%\subsection{Formal notion of Safety}
%\label{sec:safetyf}

\paragraph{Safety properties}
}

%\vspace{1mm}
%\noindent
%\textit{Safety properties.}
\begin{definition}[\cite{AS85,Lyn96}]
\label{def:pc}
A \emph{property} $\mathcal{P}$ is a set of (transactional) histories.
A property $\mathcal{P}$ is a \emph{safety} property if it satisfies:
%the following two conditions:
\begin{enumerate}
\item \emph{Prefix-closure:}
every prefix $H'$ of a history $H \in \mathcal{P}$
is also in $\mathcal{P}$ and
\item \emph{Limit-closure:}
for any infinite sequence of finite histories
$H^0,H^1,\ldots $ such that for all $i$, $H^i \in \mathcal{P}$ and $H^i$ is a prefix of $H^{i+1}$,
the infinite history that is the \emph{limit} of the sequence is also in $\mathcal{P}$.
\end{enumerate}
\end{definition}
Notice that the set of histories produced by a TM implementation $M$
is prefix-closed. Therefore, every infinite history of $M$
is the limit of an infinite sequence of ever-extending finite
histories of $M$. Thus, to prove that $M$ satisfies a safety property
$P$, it is enough to show that all finite histories of $M$ are in
$P$. Indeed, limit-closure of $P$ then implies that every infinite
history of $M$ is also in $P$.

\section{DU-Opacity}
\label{sec:kr}
In this section, we introduce our correctness criterion,
\emph{du-opacity},
and prove that a restriction of it is a limit-closed property.
\begin{definition}
\label{def:comp}
Let $H$ be any history. A \emph{completion of $H$}, denoted ${\bar H}$,
is a history derived from $H$ as follows:
\begin{itemize}
\item for every incomplete t-operation $op_k$ of $T_k \in \txns(H)$ in $H$,
if $op_k=\Read_k \vee \Write_k\vee \TryA_k()$, 
insert $A_k$ somewhere after the invocation of $op_k$; 
otherwise, if $op_k=\TryC_k()$, 
insert $C_k$ or $A_k$ somewhere after the last event of $T_k$.
\item for every complete transaction $T_k \in \txns(H)$ that is not t-complete, insert $\textit{tryC}_k\cdot A_k$ after the last event of transaction $T_k$.
\end{itemize}
\end{definition}
%
%Now we define our correctness criterion.
%We begin with defining what it means for a transaction to read from
%another transaction in a t-sequential legal history.
%
Let $H$ be any history and $S$ be a legal t-complete
t-sequential history that is equivalent to some completion of $H$.
Let $<_S$ be the total order on %t-operations
transactions in $S$.

%Let $\Read_k^n(X)$ denote the $n^{th}$ t-read of $X$ performed by transaction $T_k$.
For any $\Read_k(X)$ that does not return $A_k$, let $S^{k,X}$ denote the prefix of $S$ up to the response of $\Read_k(X)$ and $H^{k,X}$ denotes the prefix of $H$ up to the response of $\Read_k(X)$.
Let $S_{H}^{k,X}$ denote the subsequence of $S^{k,X}$ derived by
removing from $S^{k,X}$ the events of all transactions $T_m \in
\txns(H)$ such that $H^{k,X}$ does not contain an invocation of
$\TryC_m()$.
We refer to $S_{H}^{k,X}$ as the \emph{local serialization} 
for $\Read_k(X)$ with respect to $H$ and $S$.
%In the rest of the paper, when we refer to a specific $\Read_k(X)$ and $S_{H}^{k,X}$, we omit the variable $n$ for succinctness.

%Consider a read operation $\Read_k(X)$ in $S$ that does not abort,
%i.e., which returns a value $v \in V$;
%its \emph{read-from transaction}, denoted $\rho_S(\Read_k(X))$,
%is defined as follows:
%\begin{enumerate}
%\item[(1)]
%If there is $\Write_k(X,v)$ performed by $T_k$ that is the latest write
%in $T_k$ such that $\Write_k(X,v) <_S \Read_k(X)$,
%then $\rho_S(\Read_k(X))=T_k$.
%I.e., if the same transaction writes to $X$, then $\Read_k(X)$
%is mapped to $T_k$.
%\item[(2)]
%Otherwise, $\rho_S(\Read_k(X))=T_m$,
%where $T_m$ is the earliest committed transaction in $S$
%that performs $\Write_m(X,v)$ such that $T_m <_S T_k$, and
%there is no committed transaction $T_i$ in $S$
%that performs $\Write_i(X,v')$, $v' \neq v$,
%such that $T_m <_S T_i <_S T_k$.
%I.e., $\Read_k(X)$ is mapped to the earliest committed transaction
%writing the value it read, which is not overwritten.
%\end{enumerate}
We are now ready to present our correctness condition, \emph{du-opacity}.
\begin{definition} %[Kuznetsov and Ravi~\cite{KR11}]
\label{def:opacityKR}
A history $H$ is \emph{du-opaque} if there is
a legal t-complete t-sequential history $S$ such that
\begin{enumerate}
\item[(1)] there exists a completion of $H$ that is equivalent to $S$, and
\item[(2)] for every pair of transactions $T_k,T_m \in \txns(H)$, if $T_k \prec_H^{RT} T_m$,
then $T_k <_S T_m$, i.e., $S$ respects the real-time ordering of
transactions in $H$, and
%for every $\Read_k(X)$ and its \emph{read-from transaction} $T_m$ in $S$, if
%$T_k \prec_{H}^{DU} T_m$ then, $T_k$ must precede $T_m$ in $S$, and
\item[(3)] each $\Read_k(X)$ in $S$ that does not return $A_k$ is
  legal in $S_{H}^{k,X}$.
% where $S_{H}^{k,X}$ is the local serialization for $\Read_k(X)$ with respect to $H$ in $S$.
%\item[(3)] for each $\Read_k(X)$ in $H$ that does not return $A_k$ such that $\rho_S(\Read_k(X))=T_m$; $\Read_k(X) \not\prec_{H}^{RT} \TryC_m()$.
%\item[(4)] if $\Read_k(X) \prec_H^{RT} T_m$, then for any $\Read_m(X)$, $\rho_{S}(\Read_k(X))=\rho_{S}(\Read_m(X))$ or $\rho_{S}(\Read_k(X)) <_S \rho_{S}(\Read_m(X))$.
\end{enumerate}
We then say that $S$ is a (du-opaque) \emph{serialization} of $H$.
Let $\ms{seq}(S)$ denote the \emph{sequence of transactions} in $S$
and $\ms{seq}(S)[k]$ denote the $k^{th}$ transaction in this sequence.
\end{definition}
Informally, a history $H$ is du-opaque if 
there exists a legal t-sequential history $S$  that is equivalent to $H$, respects the 
real-time ordering of transactions in $H$ and every t-read is legal in
its local serialization with respect to $H$ and $S$.
The third condition reflects the implementation's deferred-update semantics, i.e., 
the legality of a t-read in a serialization does not depend on transactions that start committing after the response of the t-read.

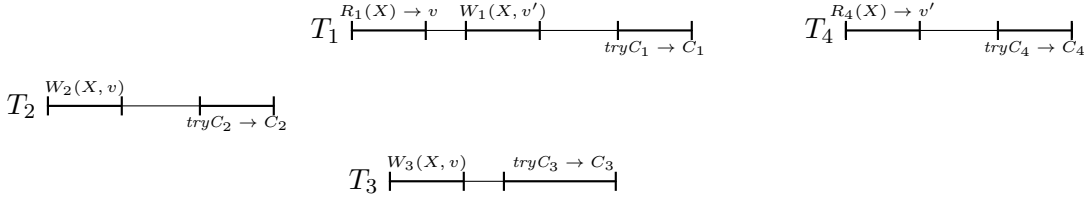
\begin{figure*}[t]
\begin{tikzpicture}
\node (r1) at (4.5,-1) [] {};
\node (w1) at (6,-1) [] {};
\node (c1) at (8,-1) [] {};

\node (r4) at (11,-1) [] {};
\node (c4) at (13,-1) [] {};

\node (w2) at (0.5,-2) [] {};
\node (c2) at (2.5,-2) [] {};

\node (w3) at (5,-3) [] {};
\node (c3) at (6.8,-3) [] {};

\draw (r1) node [above] {\tiny {$R_1(X)\rightarrow v$}};
\draw (w1) node [above] {\tiny {$W_1(X,v')$}};
\draw (c1) node [below] {\tiny {$\TryC_1 \rightarrow C_1$}};

\draw (r4) node [above] {\tiny {$R_4(X) \rightarrow v'$}};
\draw (c4) node [below] {\tiny {$\TryC_4 \rightarrow C_4$}};

\draw (w2) node [above] {\tiny {$W_2(X,v)$}};
\draw (c2) node [below] {\tiny {$\TryC_2 \rightarrow C_2$}};

\draw (w3) node [above] {\tiny {$W_3(X,v)$}};
\draw (c3) node [above] {\tiny {$\TryC_3 \rightarrow C_3$}};

\begin{scope}
\draw [-,thin] (4,-1) node[left] {$T_1$} to (8.5,-1);
\draw [|-|,thick] (4,-1) node[left] {} to (5,-1);
\draw [|-|,thick] (5.5,-1) node[left] {} to (6.5,-1);
\draw [|-|,thick] (7.5,-1) node[left] {} to (8.5,-1);

\draw [|-|,thin] (10.5,-1) node[left] {$T_4$} to (13.5,-1);
\draw [|-|,thick] (10.5,-1) node[left] {} to (11.5,-1);
\draw [|-|,thick] (12.5,-1) node[left] {} to (13.5,-1);
\end{scope}  
\begin{scope}
\draw [-,thin] (0,-2) node[left] {$T_2$} to (3,-2);
\draw [|-|,thick] (0,-2) node[left] {} to (1,-2);
\draw [|-|,thick] (2,-2) node[left] {} to (3,-2);
\end{scope}  

\begin{scope}
\draw [-,thin] (4.5,-3) node[left] {$T_3$} to (7.5,-3);
\draw [|-|,thick] (4.5,-3) node[left] {} to (5.5,-3);
\draw [|-|,thick] (6,-3) node[left] {} to (7.5,-3);
\end{scope}  
\end{tikzpicture}
\caption{A du-opaque history $H$;
there exists a serialization $S$ of $H$ such that each t-read in $S$ has a legal local serialization with respect to $H$ and $S$}
\label{fig:rf}
\end{figure*}
An example of a du-opaque history $H$ is presented in Figure~\ref{fig:rf}.
Let $S$ be the t-complete t-sequential history such that
$\ms{seq}(S)=T_2,T_3,T_1,T_4$ and $S$ is equivalent to $H$ ($H$ is its own completion).
It is easy to see that $S$ is legal and respects
the real-time order of transactions in $H$.
We now need to prove that each t-read performed in $S$ has a local serialization with respect to $H$ in $S$ that is legal.
Consider $\Read_1(X)$ in $S$; since $T_2$ is t-complete in $H^{1,X}$, it follows that $\Read_1(X)$ is legal in $T_2\cdot \Read_1(X)$ (local serialization for $\Read_1(X)$ with respect to $H$ and $S$).
Similarly, since $T_1,T_2,T_3$ are t-complete in $H^{4,X}$, $\Read_4(X)$ is legal in $T_2 \cdot T_3 \cdot T_1 \cdot \Read_4(X)$ (local serialization for $\Read_4(X)$ with respect to $H$ and $S$)
Thus, $S$ is a du-opaque serialization of $H$.

%%
%
%Now we show that du-opacity has an important property:
%every serialization of a du-opaque history yields
%a serialization for each of its prefixes that preserves the read-from
%relations of the original history intact.
For a history $H$, let $H^i$ be the finite prefix of $H$ of length $i$, i.e., consisting of the first $i$ events of $H$. Now we show a property of du-opaque histories that is going to be
instrumental in the rest of the paper.
\begin{figure*}[t]
\begin{tikzpicture}
\node (w1) at (1,0) [] {};
\node (c1) at (3.5,0) [] {};
\node (r2) at (3.5,-1) [] {};
%\node (c2) at (3.7,-1) [] {};
%\node (w3) at (1.2,-2) [] {};
%\node (r3) at (5,-2) [] {};
\node (r4) at (9.5,-2) [] {};
%\node (a4) at (10,-2) [] {};
\node (r5) at (6,-2) [] {};

\draw (w1) node [above] {\tiny {$W_1(X,1)$}};
\draw (c1) node [above] {\tiny {$\TryC_1$}};
\draw (r2) node [above] {\tiny {$R_2(X) \rightarrow 1$}};
%\draw (c2) node [above] {\tiny {$\TryC_2$}};
%\draw (w3) node [above] {\tiny {$W_3(X,1)$}};
%\draw (r3) node [below] {\tiny {}};
\draw (r4) node [above] {\tiny {$R_i(X) \rightarrow 0$}};
\draw (r5) node [above] {\tiny {$R_3(X) \rightarrow 0$}};
%\draw (a4) node [above] {\tiny {$\TryA_i() \rightarrow A_i$}};
%\draw (r5) node [above] {};

\begin{scope}   
\draw [|-|,thick] (0.5,0) node[left] {} to (1.5,0);
\draw [|-,thick] (2.5,0) node[left] {} to (9.8,0);
\draw [-,thin] (0.5,0) node[left] {$T_1$} to (9.8,0);
\draw [-,dashed] (9.8,0) node[left] {} to (11.8,0);
\end{scope}

\begin{scope}
\draw [|-|,thick] (3,-1) node[left] {$T_2$} to (4,-1);
\draw [-,dashed] (4,-1) node[left] {} to (12,-1);
%\draw [-,thick] (4,-1) node[left] {$T_2$} to (4,-1) ; 
\end{scope}  
\begin{scope}
%\draw [|-|,thick] (1,-2) node[left] {} to (2,-2);
\draw [|-|,thick] (5.5,-2) node[left] {$T_3$} to (7,-2);
\draw [-,dashed] (7.2,-2) node[left] {} to (8.2,-2);
%\draw [|-|,dashed] (6.5,-2) node[left] {$T_i$} to (11,-2);
\draw [|-|,thick] (9,-2) node[left] {$T_i$} to (10.5,-2);
%\draw [|-|,thick] (9,-2) node[left] {} to (11,-2);
%\draw [|-|,dashed] (.5,-2) node[left] {} to (9.5,-2);
\draw [->,dashed] (12,-2) node[above] {$\rightarrow \infty$} to (14,-2);
%\draw [-,thick] (1,-2) node[left] {$T_3$} to (6,-2) ; 
\end{scope}  
\end{tikzpicture}
\caption{Each finite prefix of the history is du-opaque,
but the infinite limit of the ever-extending sequence is not du-opaque}
\label{fig:op-example}
\end{figure*}
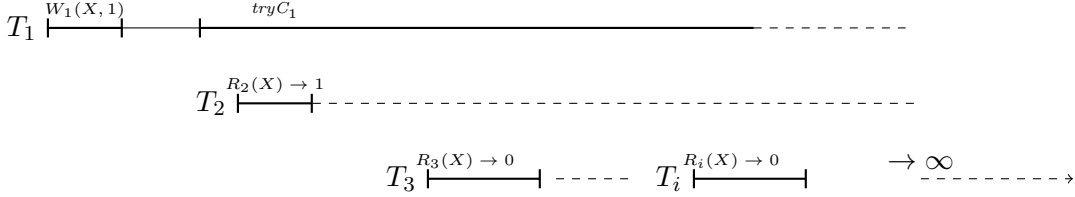
\begin{lemma}
\label{lm:dusep}
Let $H$ be a du-opaque history and $S$ be a serialization of $H$.
For any $i\in \Nat$, there exists a serialization $S^i$ of $H^i$, such that
$\ms{seq}(S^i)$ is a subsequence of $\ms{seq}(S)$.
%(2) for every $\Read_k(X)$ in $H^i$ that does not return $A_k$,
%if $\rho_{S}(\Read_k(X))=T_m$, then $\rho_{S^i}(\Read_k(X))=T_m$.
\end{lemma}
\begin{proof}
Given $H$, $S$ and $H^i$,
we construct a t-complete t-sequential history $S^i$ as follows:
\begin{itemize}
\item for every transaction $T_k$ that is t-complete in $H^i$,  $S^i|k=S|k$.
\item for every transaction $T_k$ that is complete but not t-complete in $H^i$,
    $S^i|k$ consists of the sequence of events in $H^i|k$,
    immediately followed by $\TryC_k()\cdot A_k$.
\item for every transaction $T_k$ with an incomplete
t-operation $op_k=\Read_k \vee \Write_k \vee \TryA_k()$ in $H^i$, $S^i|k$ is the
sequence of events in $S|k$ up to the invocation of $op_k$, immediately
followed by $A_k$.
\item for every transaction $T_k \in \txns(H^i)$ with an incomplete
t-operation $op_k=\TryC_k()$, $S^i|k=S|k$.
\end{itemize}
By the above construction, $S^i$ is indeed a t-complete history 
and every transaction that appears in $S^i$ also appears in $S$.
Now we order transactions in $S^i$ so that $\ms{seq}(S^i)$ is a subsequence of $\ms{seq}(S)$.

Note that $S^i$ is derived from events contained in some completion $\bar H$ of
$H$ that is equivalent to $S$.
Since $S^i$ contains events from every complete t-operation in $H^i$ and
other events included are borrowed from $\bar H$, there exists a completion of $H^i$ that is equivalent to $S^i$.

We now prove that $S^i$ is a serialization of $H^i$.
First we observe that $S^i$ respects the real-time order of $H^i$.
Indeed, if $T_j \prec_{H^i}^{RT} T_k$, then $T_j \prec_{H}^{RT} T_k$
and $T_j <_{S} T_k$. Since $\ms{seq}(S^i)$ is a subsequence of
$\ms{seq}(S)$, we have  $T_j <_{S^i} T_k$.

%Note that since $\ms{seq}(S^i)$ is a prefix of $\ms{seq}(S)$ and $S$ is legal, for any $\Read_k(X)$ in $H^i$, if $\rho_S(\Read_k(X))=T_m$, then $\rho_{S^i}(\Read_k(X))=T_m$.
To show that $S^i$ is legal,
suppose, by way of contradiction, 
that there is some $\Read_k(X)$ that returns $v\neq A_k$ in $H^i$ 
such that $v$ is not the latest written value of $X$ in $S^i$.
If $T_k$ contains a $\Write_k(X,v')$ preceding $\Read_k(X)$ such that
$v\neq v'$ and $v$ is not the latest written value for $\Read_k(X)$
in $S^i$, it is also not the latest written value  for $\Read_k(X)$
in $S$, which is a contradiction.
Thus, the only case to consider is when $\Read_k(X)$ should return a value
written by another transaction.

Since $S$ is a serialization of $H$,
there exists a committed transaction
$T_m$ that performs the last $\Write_m(X,v)$ that precedes
$\Read_k(X)$ in $T_k$ in $S$.
%; $v$ is latest written value of $X$ in $S$ and
Moreover, since $\Read_k(X)$ is legal in the local serialization for
$\Read_k(X)$ in $H$ with respect to $S$,
the prefix of $H$ up to the response of $\Read_k(X)$ must contain an invocation of $\TryC_m()$.
Thus, $\Read_k(X) \not\prec_H^{RT} \TryC_m()$  and $T_m\in \txns(H^i)$.
By construction of $S^i$, $T_m \in \txns(S^i)$ and $T_m$ is committed in $S^i$.

We have assumed, towards a contradiction, that
$v$ is not the latest written value for $\Read_k(X)$  in $S^i$.
Hence, there exists a committed transaction $T_j$ that performs
$\Write_j(X,v');v' \neq v$ in $S^i$ such that $T_m <_{S^i} T_j <_{S^i} T_k$.
%This implies that $\Write_m(X,v)$ is not the \emph{latest write} in
%$S^i$  such that $T_m <_{S^i} T_k$ and $T_m$ is committed in $S^i$.
But this is not possible since $\ms{seq}(S^i)$ is a subsequence of
$\ms{seq}(S)$.

Thus, $S^i$ is a legal t-complete t-sequential history equivalent to
some completion of $H^i$.
Now, by the construction of $S^i$, for every $\Read_k(X)$ that does not
return $A_k$ in $S^i$, we have ${S^i}_{H^i}^{k,X}={S}_{H}^{k,X}$.
Indeed, the transactions that appear before $T_k$ in
${S^i}_{H^i}^{k,X}$ are those with a $\TryC$ event before the
response of  $\Read_k(X)$ in $H$ and are committed in $S$.
Since $\ms{seq}(S^i)$ is a subsequence of $\ms{seq}(S)$, we have
${S^i}_{H^i}^{k,X}={S}_{H}^{k,X}$.
Thus, $\Read_k(X)$ is legal in  ${S^i}_{H^i}^{k,X}$.
%If there is $\Write_k(X,v)$ performed by $T_k$
%such that $\Write_k(X,v) \prec_H^{RT}
%\Read_k(X)$, the claim is immediate.
%Suppose no such $\Write_k(X,v)$ is performed by $T_k$.
%From arguments used in (2), it follows that
%there exists a committed transaction
%$T_m$ that performs $\Write_m(X,v)$ such that
%$T_m <_S T_k$; $v$ is latest written value of $X$ in $S$ and $\Read_k(X) \not\prec_H^{RT} \TryC_m()$.
%Thus, ${H^i}^{k,X}$ must contain the invocation of $\TryC_m()$ and since ${S^i}_{H^i}^{k,X}$ is a subsequence of $S$,
%$\Read_k(X)$ is legal in ${S^i}_{H^i}^{k,X}$.
%%
%%We can construct a serialization $S'$ of $H'$ from $S$ such that $\ms{seq}(S')$ is a subsequence of $\ms{seq}(S)$ as described in the proof of Theorem~\ref{th:pc}.
\end{proof}
%
%\subsection{DU-Opacity: Safety properties}
%\label{sec:saf}
%
Lemma~\ref{lm:dusep} implies that every prefix of a
du-opaque history has a du-opaque serialization and thus:
\begin{corollary}
\label{cr:pc}
DU-Opacity is a prefix-closed property.
\end{corollary}
We show, however,  that du-opacity is, in general, not limit-closed.
We present an infinite history that is not du-opaque,
but every its prefix is.
\begin{proposition}\label{lm:lmn}
DU-Opacity is not a limit-closed property.
\end{proposition}
\begin{proof}
Let $H^j$ denote a finite prefix of $H$ of length $j$.
Consider an infinite history $H$ that is the limit of the
histories $H^j$ defined as follows (see Figure~\ref{fig:op-example}):
\begin{itemize}
\item[--]
Transaction $T_1$ performs a $\Write_1(X,1)$ and then
invokes $\TryC_1()$ that is incomplete in $H$.
\item[--]
Transaction $T_2$ performs a $\Read_2(X)$ that overlaps with $\TryC_1()$
and returns $1$.
\item[--]
There are infinitely many transactions $T_i$, $i\geq 3$,
each of which performing a single $\Read_i(X)$ that returns $0$
such that each $T_i$ overlaps with both $T_1$ and $T_2$.
\end{itemize}

A t-complete t-sequential history $S^j$ is derived from the
sequence $T_3, \ldots, T_j,T_0,T_1$ in which
(1) $\TryC_1()$ is completed by inserting $C_1$
immediately after its invocation
and (2) any incomplete $\Read_j(X)$ is completed by inserting $A_j$ immediately after its invocation.
It is easy to observe that $S^j$ is indeed a serialization of $H^j$.

However, there is no serialization of $H$.
Suppose that such a serialization $S$ exists.
Since every transaction that participates in $H$ must participate in $S$,
there exists $n \in \Nat$ such that $\ms{seq}(S)[n]=T_1$.
Consider the transaction at index $n+1$, say $T_i$ in $\ms{seq}(S)$.
But for any $i\geq 3$, $T_i$ must precede $T_1$ in any serialization (by legality),
which is a contradiction.
\end{proof}
We next prove that du-opacity is limit-closed if we assume that, in
an infinite history, every transaction eventually completes (but not
necessarily t-completes).

The proof uses K\"{o}nig's Path Lemma on a rooted directed graph, $G$.
Let $v_0$ be the root vertex of $G$.
We say that $v_k$, a vertex of $G$, is \emph{reachable} from $v_0$,
if there is a sequence of vertices $v_0 \ldots, v_k$ such that for each $i$,
there exists an edge from $v_{i}$ to $v_{i +1}$.
%The \emph{level} of a vertex $v$ in $G$ is the \emph{length of the longest walk} from the root to $v$.
%We say that $G$ is \emph{connected} if for each level $i$, there exists a vertex \emph{reachable} from the \emph{root} of $G$.
$G$ is \emph{connected} if every vertex in $G$ is reachable from $v_0$.
$G$ is \emph{finitely branching} if every vertex in $G$ has a finite out-degree.
$G$ is \emph{infinite} if the set of vertices in $G$ is not finite.
\begin{lemma}[K\"{o}nig's Path Lemma~\cite{konig}]
\label{lm:konig}
If $G$ is an infinite connected finitely branching rooted directed graph,
then $G$ contains an infinite sequence of vertices $v_0,v_1, \ldots $
such that $v_0$ is the \emph{root},
for every $i \geq 0$, there is an edge from $v_i$ to $v_{i+1}$,
and for every $i \neq j$, $v_i \neq v_j$.
\end{lemma}
We first prove the following lemma concerning du-opaque serializations.

For a transaction $T\in \txns(H)$, we define the \emph{live set of $T$ in $H$}, denoted $\ms{Lset}_H(T)$ ($T$ included) as follows:
every transaction $T' \in \txns(H)$ such that neither the last event of $T'$ precedes the first event of $T$ in
$H$ nor the last event of $T$ precedes the first event of $T'$ in $H$ is contained in $\ms{Lset}_H(T)$.
We say that transaction $T'\in \txns(H)$ \emph{succeeds the live set of $T$} and we write $T\prec_H^{LS} T'$ if in $H$, for all $T''\in
\ms{Lset}_H(T)$, $T''$ is complete and the last event of $T''$
precedes the first event of $T'$.
\begin{lemma}
\label{lm:full}
Let $H$ be a finite du-opaque history and assume $T_k \in \txns(H)$ be a complete transaction in $H$ such that every transaction in $\ms{Lset}_H(T_k)$ is complete in $H$.
Then there exists a serialization $S$ of $H$
such that for all $T_k, T_m \in\txns(H)$; $T_k\prec_H^{LS} T_m$, we have $T_k<_S T_m$.
%such that $T <_S T''$, where $T''\in \txns(H)$ is any transaction whose first event succeeds the last event of each $T'\in \ms{Lset}_H(T)$ in $H$.
\end{lemma}
\begin{proof}
Since $H$ is du-opaque, there exists a serialization ${\tilde S}$ of $H$.

Let ${ S}$ be a t-complete t-sequential history such that $\txns(\tilde S)=\txns(S)$, and $\forall~T_i \in \txns(\tilde S): S|i={\tilde S}|i$.
%For each $(T_k, T_m)$; $T_k\prec_H^{LS} T_m$ and $T_m<_{\tilde S} T_k$ (observe that $T_k$ must be complete, but not t-complete),
We now perform the following procedure iteratively to derive $\ms{seq}(S)$ from $\ms{seq}(\tilde S)$.
Initially $\ms{seq}(S)=\ms{seq}(\tilde S)$.
For each $T_k \in \txns(H)$, let $T_{\ell}\in \txns(H)$ denote the earliest transaction in ${\tilde S}$ such that $T_k \prec_H^{LS} T_{\ell}$.
If $T_{\ell} <_{\tilde S} T_k$ (implying $T_k$ is not t-complete), then move $T_k$ to immediately precede $T_{\ell}$ in $\ms{seq}(S)$.
%informally, $T_k$ is moved immediately behind $T_{\ell}$ in $S$ while the events of each transaction in $S$ are identical to its events in ${\tilde S}$.
%
%\begin{itemize}
%\item if $T_{\ell} <_{\tilde S} T_{k}$, then $T_k <_S T_{\ell}$, and
%\item $\forall~T_i,T_j \in \txns(\tilde S)$;$i\neq k, j\neq k$, $T_i <_{\tilde S} T_j$, then $T_i <_{S} T_j$, and
%\item $\forall~T_i \in \txns(\tilde S); T_i <_{\tilde S} T_{\ell}$, then $T_i <_S T_k$, and
%\item $\forall~T_i \in \txns(\tilde S); T_{\ell} <_{\tilde S} T_{i}$, then $T_k <_S T_i$.
%\end{itemize}
%%

By construction, $S$ is equivalent to ${\tilde S}$ and for all $T_k, T_m \in\txns(H)$; $T_k\prec_H^{LS} T_m$, $T_k<_{S} T_m$
We claim that $S$ is a serialization of $H$.
%For each $T_k \in \txns(H)$, let $\ms{Move}_S^{\tilde S}(T_k)$ denote the set of transactions that precede $T_k$ in ${\tilde S}$, but succeed $T_k$ in $S$.
%By construction, for any $(T_k, T_m)$; $T_k\prec_H^{LS} T_m$ and any $T_i \in \txns(H)$, $T_i \prec_H^{RT} T_k$, $T_i <_S T_k$.
Observe that any two transactions that are complete in $H$, but not t-complete are not related by real-time order in $H$.
By construction of $S$, for any transaction $T_k\in \txns(H)$, the set of transactions that precede $T_k$ in ${\tilde S}$, but succeed $T_k$ in $S$ are not related to $T_k$ by real-time order.
Since $\tilde S$ respects the real-time order in $H$, this holds also
for $S$.
%For any $(T_k,T_{\ell})$; $T_k \prec_H^{LS} T_{\ell}$, if $T_{\ell} <_{\tilde S} T_k$, then $T_k$ is not t-complete in $H$. Thus,
%moving $T_k$ to immediately precede $T_{\ell}$ does not affect the real-time order relations in $S$.

We now show that $S$ is legal.
Consider any $\Read_k(X)$ performed by some transaction $T_k$ that returns $v\in V$
in $S$ and let $T_{\ell}\in \txns(H)$ be the earliest transaction in ${\tilde S}$ such that $T_k \prec_H^{LS} T_{\ell}$.
%Suppose otherwise that $v$ is not the latest written value of $X$ in $S$.
Suppose, by contradiction, that $\Read_k(X)$ is not legal in $S$.
Thus, there exists a committed transaction $T_m$ that performs
$\Write_m(X,v)$ in ${\tilde S}$ such that $T_m=T_{\ell}$ or $T_{\ell} <_{\tilde S} T_m <_{\tilde S} T_k$.
Note that, by our assumption, $\Read_k(X) \prec_H^{RT} \TryC_{\ell}()$.
Since $\Read_k(X)$ must be legal in the local serialization of $\tilde S$ with respect to $H$, $\Read_k(X) \not\prec_H^{RT} \TryC_m()$.
Thus, $T_m \in \ms{Lset}_H(T_k)$. Therefore $T_m\neq T_{\ell}$.
Moreover, $T_m$ is complete, and since it commits in $\tilde S$, it
is also t-complete in $H$ and the last event of $T_m$
precedes the first event of $T_{\ell}$ in $H$, i.e., $T_m\prec_H^{RT}
T_{\ell}$. Hence, $T_{\ell}$ cannot precede $T_m$ in ${\tilde S}$---a contradiction.

%Now consider any $\Read_i(X)$; $T_{\ell} <_S T_i$ that returns $v\in V$ in $S$.
%Observe that if $T_{\ell}$ precedes $T_k$ in ${\tilde S}$,
%then $T_k$ is not t-complete in $H$.
Observe also that since $T_k$ is complete in $H$ but not t-complete,
$H$ does not contain an invocation of $\TryC_k()$. Thus,
the legality of any other transaction is unaffected by moving $T_k$ to
precede $T_{\ell}$ in $S$.
Thus, $S$ is a legal t-complete t-sequential history equivalent to some completion of $H$.
The above arguments also prove that every t-read in $S$ is legal in
its local serialization with respect to $H$ and $S$ and, thus, $S$ is
a serialization of $H$.
\end{proof}
\begin{theorem}
\label{th:lc}
Under the restriction that in any infinite history $H$,
every transaction $T_k \in \txns(H)$ is complete,
du-opacity is a limit-closed property.
\end{theorem}
\begin{proof}
We are given an infinite sequence of finite ever-extending du-opaque
histories, let $H$ be the corresponding infinite limit history.
We want to show that $H$ is also du-opaque.
By Corollary~\ref{cr:pc}, every prefix of $H$ is du-opaque.
Therefore, we can assume the sequence of du-opaque histories
to be $H^0,H^1, \ldots H^i,H^{i+1},\ldots$, where each $H^i$ is the prefix of
$H$ of length $i$.

We construct a rooted directed graph $G_H$ as follows:
%We say that a transaction $T_k$ is \emph{t-complete} in $H^i$ if $H^i$ contains the last event of $T_k$ in the infinite history $H$ (by assumption, this is a response event).
%Let $\ms{cseq}(S^i)$ denote the subsequence of $\ms{seq}(S^j)$ reduced to \emph{t-complete transactions} in $H^i$.
%
\begin{enumerate}
\item[(0)]
The root vertex of $G_H$ is ($H^0, S^0)$ where $S^0$ and $H^0$ contain the initial transaction $T_0$.
\item[(1)]
Each non-root vertex of $G_H$ is a tuple $({H}^i, S^i)$,
where $S^i$ is a du-opaque serialization of ${H}^i$ that satisfies
the condition specified in Lemma~\ref{lm:full}: for all $T_k, T_m \in\txns(H)$; $T_k\prec_H^{LS} T_m$, $T_k<_S T_m$.
Note that there exist several possible serializations for any $H^i$.
For succinctness, in the rest of this proof, when we refer to a specific $S^i$,
it is understood to be associated with the prefix $H^i$ of $H$.
\item[(2)]
We say that a transaction $T$ is \emph{complete in $H'$ with respect to $H$}, where $H$ is any extension of $H'$
if last step of $T$ in $H$ is a response event and it is contained in $H'$.

Let $\ms{cseq}_i(S^j)$, $j\geq i$, denote the subsequence of
$\ms{seq}(S^j)$ reduced to transactions that are \emph{complete in $H^i$ with respect to $H$}.
For every pair of vertices $v=({H}^i,S^i)$ and $v'=({H}^{i+1}, S^{i+1})$ in $G_H$,
there is an edge from $v$ to $v'$ if $\ms{cseq}_i(S^i)=\ms{cseq}_i(S^{i+1})$.
\end{enumerate}
The out-degree of a vertex $v=(H^i,S^i)$ in $G_H$ is defined by
the number of possible serializations of ${H}^{i+1}$,
bounded by the number of possible permutations
of the set $\txns(S^{i+1})$, implying that $G_H$ is \emph{finitely branching}.

By Lemma~\ref{lm:dusep},
given any serialization $S^{i+1}$ of $H^{i+1}$,
there exists a serialization $S^i$ of $H^i$
such that $\ms{seq}(S^i)$ is a subsequence of $\ms{seq}(S^{i+1})$.
Indeed, the serialization $S^i$ of $H^i$ also respects the restriction specified in Lemma~\ref{lm:full}.
Since $\ms{seq}(S^{i+1})$ contains every complete transaction that takes its last step in $H$ in $H^i$,
$\ms{cseq}_i(S^i)=\ms{cseq}_i(S^{i+1})$.
Therefore, for every vertex $({H}^{i+1},S^{i+1})$,
there is a vertex $({H}^{i},S^i)$ such that $\ms{cseq}_i(S^i)={cseq}_i(S^{i+1})$.
Thus, we can iteratively construct a path from
$(H^0,S^0)$ to every vertex $(H^i,S^i)$ in $G_H$,
implying that $G_H$ is \emph{connected}.

We now apply K\"{o}nig's Path Lemma to $G_H$.
Since $G_H$ is an infinite connected finitely branching rooted directed graph,
we can derive an infinite sequence of non-repeating vertices
\[
\mathcal{L}=(H^0,S^0), (H^1, S^1), \ldots, (H^i,S^i), \ldots
\]
such that $\ms{cseq}_i(S^i)=\ms{cseq}_i(S^{i+1})$.

The rest of the proof explains how to use $\mathcal{L}$ to construct a serialization of $H$.
We begin with the following claim concerning $\mathcal{L}$.
\begin{claim}
\label{cl:lclaim}
For any $j>i$, $\ms{cseq}_i(S^i)=\ms{cseq}_i(S^{j})$.
\end{claim}
\begin{proof}
%From $\mathcal{L}$, we have the following relations:
Recall that $\ms{cseq}_i(S^i)$ is a prefix
of $\ms{cseq}_i(S^{i+1})$, and
$\ms{cseq}_{i+1}(S^{i+1})$ is a prefix of $\ms{cseq}_{i+1}(S^{i+2})$.
Also, $\ms{cseq}_i(S^{i+1})$ is a subsequence of $\ms{cseq}_{i+1}(S^{i+1})$.
%since if a transaction takes its last step of $H$ in $H^{i}$, is also t-complete in $H^i$.
Hence, $\ms{cseq}_i(S^{i})$ is a subsequence of
$\ms{cseq}_{i+1}(S^{i+2})$.
But, $\ms{cseq}_{i+1}(S^{i+2})$ is a subsequence of
$\ms{cseq}_{i+2}(S^{i+2})$.
Thus, $\ms{cseq}_i(S^i)$ is a subsequence of
$\ms{cseq}_{i+2}(S^{i+2})$.
Inductively, for any $j>i$, $\ms{cseq}_i(S^i)$ is a subsequence of
$\ms{cseq}_{j}(S^{j})$.
But $\ms{cseq}_{i}(S^{j})$ is the
subsequence  of $\ms{cseq}_{j}(S^{j})$ reduced to
transactions that are complete in $H^i$ with respect to $H$.
Thus, $\ms{cseq}_i(S^i)$ is indeed equal to $\ms{cseq}_{i}(S^{j})$.
\end{proof}
Let $f:\Nat \rightarrow \txns(H)$ be defined as follows:
$f(1) = T_0$,
For every integer $k > 1$, let
\[ i_k = \min \{ \ell \in \Nat | \forall j>\ell:
\ms{cseq}_{\ell}(S^{\ell})[k]=\ms{cseq}_j(S^j)[k]\} \]
Thus, $f(k)=\ms{cseq}_{i_k}(S^{i_k})[k]$.
%
%%function be defined in Figure~\ref{fig:func}.
%%
%\begin{figure*}[t]
%\begin{framed}
%\hspace{4cm} $f(1)=T_0$ \\
%$\forall k \in \Nat \setminus \{1\}:f(k)=\ms{cseq}_i(S^i)[k]$;$i=\ms{min} \{\ell \in \Nat | \forall j>\ell:\ms{cseq}_{\ell}(S^{\ell})[k]=\ms{cseq}_j(S^j)[k]\}$
%\end{framed}
%\caption{Function $f:\Nat \rightarrow \txns(H)$}
%\label{fig:func}
%\end{figure*}
%
\begin{claim}
\label{cl:bij}
The function $f$ is \emph{total} and \emph{bijective}.
\end{claim}
\begin{proof}
\textit{(Totality and surjectivity)}

Since each transaction $T \in \txns(H)$ is complete in some prefix $H^i$ of $H$,
for each $k\in \Nat$, there exists $i \in \Nat$ such that $\ms{cseq}_i(S^i)[k]=T$.
By Claim~\ref{cl:lclaim},
for any $j>i$, $\ms{cseq}_i(S^i)=\ms{cseq}_i(S^{j})$.
Since a transaction that is complete in $H^i$ w.r.t $H$ is also
complete in $H^j$ w.r.t $H$,
it follows that for every $j >i$, $\ms{cseq}_j(S^j)[k']=T$, with $k' \geq k$.
By construction of $G_H$ and the assumption that each transaction is complete in $H$, there exists $i\in \Nat$
such that each $T \in \ms{Lset}_{H^i}(T)$ is complete in $H^i$ with respect to $H$ and
$T$ precedes in $S^i$ every transaction whose first event succeeds the last event of each $T'\in \ms{Lset}_{H^i}(T)$ in $H^i$.
Indeed, this implies that for each $k\in \Nat$, there exists $i \in \Nat$ such that $\ms{cseq}_i(S^i)[k]=T$; $\forall j>i:
\ms{cseq}_j(S^j)[k]=T$.
%Suppose that there exists a transaction $T_m \in \txns(H)$ such that $\not\exists i\in \Nat:\ms{cseq}_i(S^i)[k]=T_m$ and $\forall j > i:\ms{cseq}_j(S^j)[k]=T_m$.
%This implies that $\forall i\in \Nat$:$\not\exists k\in \Nat,\ms{cseq}_i(S^i)[k]=T_m$ such that $\forall j > i$, $\ms{cseq}_j(S^j)[k]=T_m$ i.e. for each $j >i$ and any $i \in \Nat$: $\exists T_{\ell} \in \txns(H)$: $\ms{cseq}_j(S^j)[k]=T_{\ell}$---contradiction.
%
%\textit{(Surjectivity)}

This shows that for every $T \in \txns(H)$,
there are $i,k\in \Nat$; $\ms{cseq}_i(S^i)[k]=T$,
such that for every $j > i$, $\ms{cseq}_j(S^j)[k]=T$.
Thus, for every $T \in \txns(H)$, there is $k$ such that $f(k)=T$.

\textit{(Injectivity)}
%By definition, $\forall k \in \Nat \setminus \{1\}$, $f(k)$ is defined to be the transaction at index $k$ in $\ms{cseq}_i(S^i)$, where $i$ is the smallest natural number such that $\forall j>i:\ms{cseq}_i(S^i)[k]=\ms{cseq}_j(S^j)[k]$.

%Assume, by way of contradiction, that for some $k \neq m$, $f(k)=f(m)$.
If $f(k)$ and $f(m)$ are transactions at indices $k$, $m$
of the same $\ms{cseq}_i(S^i)$, then clearly $f(k) = f(m)$ implies $k=m$.
Suppose $f(k)$ is the transaction at index $k$ in some $\ms{cseq}_i(S^i)$
and $f(m)$ is the transaction at index $m$ in some $\ms{cseq}_{\ell}(S^{\ell})$.
For every $\ell > i$ and $k<m$,
if $\ms{cseq}_i(S^i)[k]=T$, then $\ms{cseq}_{\ell}(S^{\ell})[m] \neq T$
since $\ms{cseq}_i(S^i)=\ms{cseq}_i(S^{\ell})$.
If $\ell > i$ and $k>m$, it follows from the definition that $f(k)\neq f(m)$.
Similar arguments for the case when $\ell < i$ prove that if $f(k)= f(m)$, then $k=m$.
\end{proof}
By Claim~\ref{cl:bij}, $\mathcal{F}=f(1),f(2),\ldots, f(i) ,\ldots$
is an infinite sequence of transactions.
Let $S$ be a t-complete t-sequential history such that
$\ms{seq}(S)=\mathcal{F}$
and for each t-complete transaction $T_k$ in $H$, $S|k=H|k$; and for
transaction that is complete, but not t-complete in $H$, $S|k$ consists of the sequence of
events in $H|k$, immediately followed by $\TryA_k()\cdot A_k$.
Clearly, there exists a completion of $H$ that is equivalent to $S$.

Let $\mathcal{F}^i$ be the prefix of $\mathcal{F}$ of length $i$,
and ${\widehat S^i}$ be the prefix of $S$ such that
$\ms{seq}({\widehat S^i})=\mathcal{F}^i$.
\begin{claim}
\label{cl:final}
Let ${\widehat H^j}_i$ be a subsequence of $H^j$ reduced
to transactions in ${\widehat S^i}$ such that
each $T_k \in \txns({\widehat S^i})$ is complete in $H^j$ with respect to $H$.
Then, for every $i$, there is $j$ such that
${\widehat S^i}$ is a serialization of ${\widehat H^j}_i$.
\end{claim}
\begin{proof}
%
%Let $H^j$ denote the shortest prefix of $H$ such that every transaction $T_k \in {\widehat S^i}$ is t-complete in $H^j$.
%From $\mathcal{L}$, we obtain the serialization $S^j$ of $H^j$.
%Note that $\txns({\widehat S^i}) \subseteq \txns(S^j)$. Thus, ${\widehat S^i}$ is a subsequence of $S^j$.
Let $H^j$ be the shortest prefix of $H$ (from $\mathcal{L}$)
such that for each $T\in \txns({\widehat S^i})$,
if $\ms{seq}(S^j)[k]=T$, then for every $j'>j$, $\ms{seq}(S^{j'})[k]=T$.
From the construction of $\mathcal{F}$, such $j$ and $k$ exist.
Also, we observe that $\txns({\widehat S^i}) \subseteq \txns(S^j)$
and $\mathcal{F}^i$ is a subsequence of $\ms{seq}(S^j)$.
Using arguments similar to the proof of Lemma~\ref{lm:dusep},
it follows that ${\widehat S^i}$ is indeed a serialization of ${\widehat H^j}_i$.
\end{proof}
Since $H$ is complete, there is exactly one completion of $H$, where
each transaction $T_k$ that is not t-complete in $H$ is completed with
$\textit{tryC}_k\cdot A_k$ after its last event.
By Claim~\ref{cl:final}, the limit t-sequential t-complete history is
equivalent to this completion, is legal,
respects the real-time order of $H$, and ensures that
every read is legal in the corresponding local serialization.
Thus, $S$ is a serialization of $H$.
%$S$ is a legal t-sequential history equivalent to $H$ ($H$ is its own completion)
%such that for any $\Read_k(X)$ in $H$, if $\rho_S(\Read_k(X))=T_m$, then $\Read_k(X) \prec_H^{RT} \TryC_m()$.
%Also, for
\end{proof}
From Theorem~\ref{th:lc}, it follows that:
\begin{corollary}
\label{cr:safetytm}
Let $M$ be any TM implementation that ensures that in every infinite history $H$ of $M$, each transaction $T\in \txns(H)$ is complete in $H$.
Then, $M$ is du-opaque \emph{iff} every finite history of $M$ is du-opaque.
\end{corollary}
\section{Comparison with Other TM Consistency Definitions}
\label{sec:gkopacity}
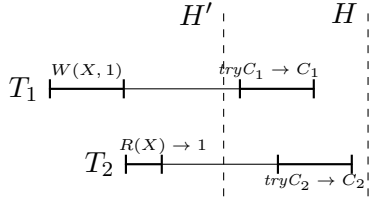
\begin{figure}[t]
\begin{tikzpicture}
\node (w2) at (2,0) [] {};
\node (r1) at (3,-1) [] {};
\node (c2) at (4.4,0) [] {};
\node (c1) at (5,-1) [] {};
%\node (n) at (2,2) [] {};
\draw (c1) node [below] {\tiny {$\TryC_2 \rightarrow C_2$}}
(r1) node [above] {\tiny {$R(X)\rightarrow 1$}};
   
\draw (c2) node [above] {\tiny {$\TryC_1 \rightarrow C_1$}}
   (w2) node [above] {\tiny {$W(X,1)$}};

%\draw (w3) node [below] {\tiny {$W(X_5)$}};   

%\draw (n) node[below] {Initial state of the linked-list: $\{1,2,3\}$} to (3,2);

\begin{scope}   
\draw [|-|,thick] (1.5,0) node[left] {} to (2.5,0);
\draw [|-|,thick] (4,0) node[left] {} to (5,0);
\draw [-,thin] (1.5,0) node[left] {$T_1$} to (5,0);
\end{scope}
\begin{scope}
\draw [|-|,thick] (2.5,-1) node[left] {} to (3,-1);
\draw [|-|,thick] (4.5,-1) node[left] {} to (5.5,-1);
\draw [-,thin] (2.5,-1) node[left] {$T_2$} to (5.5,-1) ; 
\end{scope}  
\begin{scope}
\draw [-,dashed] (3.8,1) node[left] {$H'$} to (3.8,-1.5) ; 
\end{scope}  
\begin{scope}
\draw [-,dashed] (5.7,1) node[left] {$H$} to (5.7,-1.5) ; 
\end{scope}  
\end{tikzpicture}
\caption{History $H$ is final-state opaque, while its prefix $H'$ is not final-state opaque}
\label{fig:one}
\end{figure}
\subsection{Relation to Opacity}
\label{sec:gko}
In this section, we relate du-opacity with opacity, as defined by
Guerraoui and Kapalka~\cite{tm-book}.
Note that the definition presented in~\cite{tm-book} applies to any object
with a sequential specification.
For the sake of comparison, we restrict it here to TMs with read-write
semantics.
%}

%Let $H$ be any finite history.
%Recall that a completion $H$ is any history $\bar{H}$ that
%conforms to Definition~\ref{def:comp}.
%such that:
%\begin{itemize}
%\item $H$ is a prefix of ${\bar H}$
%\item for every incomplete t-operation $op_k$ of $T_k$,
%if $op_k=\Read_k \vee \Write_k$, insert $A_k$ somewhere after the invocation of $op_k$; else
%if $op_k=\TryA_k() \vee \TryC_k()$, insert a matching response to $op_k$ somewhere after the invocation of $op_k$.
%\item for every complete transaction $T_k$ that is not t-incomplete, insert $\textit{tryC}_k\cdot A_k$ somewhere after the last event of transaction $T_k$.
%\end{itemize}
%
\begin{definition} [Guerraoui and Kapalka~\cite{GK08-opacity,tm-book}]
\label{def:opacity GK}
A finite history $H$ is \emph{final-state opaque} if there
is a legal t-complete t-sequential history $S$,
such that
\begin{enumerate}
\item[(1)] for any two transactions $T_k,T_m \in \txns(H)$,
if $T_k \prec_H^{RT} T_m$, then $T_k <_S T_m$, and
\item[(2)] $S$ is equivalent to a completion of $H$
    (cf.~Definition~\ref{def:comp}).
\end{enumerate}
%$\prec_H^{RT} \subseteq \prec_S^{RT}$.
%for any two transactions $T_k,T_m \in \txns(H)$,
%if $T_k \prec_H^{RT} T_m$, then $T_k \prec_S T_m$ and
We say that $S$ is a \emph{final-state serialization} of $H$.
\end{definition}

Figure~\ref{fig:one} presents a t-complete sequential history $H$,
demonstrating that final-state opacity is not a prefix-closed property.
$H$ is final-state opaque,
with $T_1\cdot T_2$ being a legal t-complete t-sequential history equivalent to $H$.
Let $H'=\Write_1(X,1), \Read_2(X)$ be a prefix of $H$ in which
$T_1$ and $T_2$ are t-incomplete.
By Definition~\ref{def:comp}, $T_i$ ($i=1,2$) is completed by inserting $\textit{tryC}_i\cdot
A_i$ immediately after the last event of $T_i$ in $H$.
Observe that neither $T_1\cdot T_2$ nor $T_2\cdot T_1$ are sequences
that allow us to derive a serialization of $H'$ (we assume that the initial value of $X$ is $0$).

A restriction of final-state opacity, which we refer to as \emph{opacity},
was presented in \cite{tm-book} by filtering out
histories that are not prefix-closed.
\begin{definition}[Guerraoui and Kapalka~\cite{tm-book}]
\label{def:opaque}
A history $H$ is \emph{opaque} if and only if every finite prefix $H'$
of $H$ (including $H$ itself if it is finite) is final-state opaque.
\end{definition}
%
%Since final-state opacity is defined only for finite histories,
It can be easily seen that opacity is prefix and limit-closed,
and, thus, opacity is a safety property.
%
%\subsection{DU-Opacity and Opacity: a separation}
%\label{sec:sep}
%
\begin{proposition}
\label{pr:opgkeq}
There is an opaque history that is not du-opaque.
\end{proposition}

\begin{proof}
Consider the finite history $H$ depicted in Figure~\ref{fig:lin-example}.
To prove that $H$ is opaque, we proceed by examining every prefix of $H$.

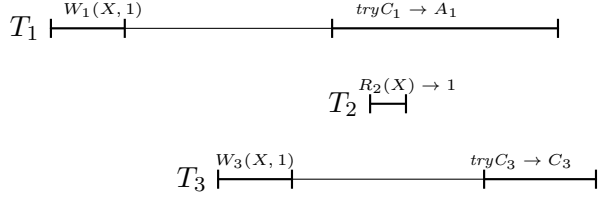
\begin{figure}[t]
\begin{tikzpicture}
\node (w1) at (-0.5,0) [] {};
\node (c1) at (3.5,0) [] {};
\node (r2) at (3.5,-1) [] {};
%\node (c2) at (3.7,-1) [] {};
\node (w3) at (1.5,-2) [] {};
\node (c3) at (5,-2) [] {};

\draw (w1) node [above] {\tiny {$W_1(X,1)$}};
\draw (c1) node [above] {\tiny {$\TryC_1 \rightarrow A_1$}};
\draw (r2) node [above] {\tiny {$R_2(X) \rightarrow 1$}};
%\draw (c2) node [above] {\tiny {$\TryC_2$}};
\draw (w3) node [above] {\tiny {$W_3(X,1)$}};
\draw (c3) node [above] {\tiny {$\TryC_3 \rightarrow C_3$}};

\begin{scope}   
\draw [|-|,thick] (-1.2,0) node[left] {} to (-0.2,0);
\draw [|-|,thick] (2.5,0) node[left] {} to (5.5,0);
\draw [-,thin] (-1.2,0) node[left] {$T_1$} to (5.5,0);
\end{scope}

\begin{scope}
\draw [|-|,thick] (3,-1) node[left] {$T_2$} to (3.5,-1);
%\draw [|-|,thick] (3.5,-1) node[left] {} to (4,-1);
%\draw [-,thick] (4,-1) node[left] {$T_2$} to (4,-1) ; 
\end{scope}  
\begin{scope}
\draw [|-|,thick] (1,-2) node[left] {} to (2,-2);
\draw [|-|,thick] (4.5,-2) node[left] {} to (6,-2);
\draw [-,thin] (1,-2) node[left] {$T_3$} to (6,-2) ; 
\end{scope}  
\end{tikzpicture}
\caption{History is opaque, but not du-opaque}
\label{fig:lin-example}
\end{figure}

\begin{enumerate}
\item
Each prefix up to the invocation of $\Read_2(X)$ is trivially final-state opaque.
\item
Consider the prefix, $H^i$ of $H$ where the $i^{th}$ event is the response of $\Read_2(X)$. Let $S^i$ be a t-complete t-sequential history derived from the sequence $T_1,T_2$ by inserting $C_1$ immediately after the invocation of $\TryC_1()$. It is easy to see that $S^i$ is a final-state serialization of $H^i$.
\item
Consider the t-complete t-sequential history $S$ derived from the sequence $T_1,T_3,T_2$ in which each transaction is t-complete in $H$. Clearly, $S$ is a final-state serialization of $H$.
\end{enumerate}
Since $H$ and every (proper) prefix of it are final-state opaque,
$H$ is opaque.

Clearly, the only final-state serialization $S$ of $H$ is specified by $\ms{seq}(S)=T_1,T_3,T_2$. Consider $\Read_2(X)$ in $S$; since $H^{2,X}$, the prefix of $H$ up to the response of $\Read_2(X)$ does not contain an invocation of $\TryC_3()$, the local serialization for $\Read_2(X)$ with respect to $H$ and $S$, $S_{H}^{2,X}$ is $T_1\cdot \Read_2(X)$. But $\Read_2(X)$ is not legal in $S_{H}^{2,X}$---contradiction.
Thus, $H$ is not du-opaque.
\end{proof}
\begin{theorem}
\label{th:gkkr}
DU-Opacity $\subsetneq$ Opacity.
\end{theorem}
\begin{proof}
We first claim that every finite du-opaque history is opaque.
Let $H$ be a finite du-opaque history.
By definition, there exists a final-state serialization $S$ of $H$.
Since du-opacity is a prefix-closed property,
every prefix of $H$ is final-state opaque.
Thus, $H$ is opaque.

Again, since
every prefix of a du-opaque history is also du-opaque,
by Definition~\ref{def:opaque},
every infinite du-opaque history is also opaque.

Proposition~\ref{pr:opgkeq} now establishes that du-opacity
is indeed a restriction of opacity.
\end{proof}
%
%
%\subsection{DU-Opacity and Opacity: Equivalence}
%\label{sec:opeqdef}
%
We now show that du-opacity is equivalent to opacity
assuming that no two transactions write identical values
to the same t-object (``unique-write'' assumption).

Let Opacity$_{ut}$ $\subseteq$ Opacity, be a property defined as follows:
\begin{enumerate}
\item[(1)] an infinite opaque history $H \in$ Opacity$_{ut}$ \emph{iff} every transaction $T\in \txns(H)$ is complete in $H$, and
\item[(2)] an opaque history $H \in$ Opacity$_{ut}$ \emph{iff} for any two transactions $T_k,T_m \in
\txns(H)$ that perform $\Write_k(X,v)$ and $\Write_m(X,v')$ respectively, $v\neq v'$.
\end{enumerate}
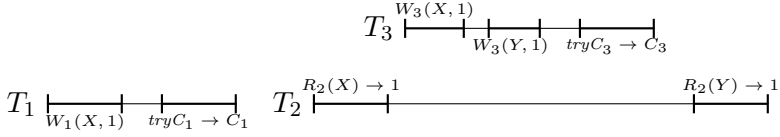
\begin{figure*}[!ht]
\begin{tikzpicture}
\node (w1) at (0,0) [] {};
\node (c1) at (1.5,0) [] {};
\node (r1) at (3.5,0) [] {};
\node (r2) at (8.5,0) [] {};

\node (w2) at (4.6,1) [] {};
\node (w22) at (5.6,1) [] {};
\node (c2) at (7,1) [] {};

\draw (w1) node [below] {\tiny {$W_1(X,1)$}};
\draw (c1) node [below] {\tiny {$\TryC_1 \rightarrow C_1$}};
\draw (r1) node [above] {\tiny {$R_2(X) \rightarrow 1$}};
\draw (r2) node [above] {\tiny {$R_2(Y) \rightarrow 1$}};

\draw (w2) node [above] {\tiny {$W_3(X,1)$}};
\draw (w22) node [below] {\tiny {$W_3(Y,1)$}};
\draw (c2) node [below] {\tiny {$\TryC_3 \rightarrow C_3$}};
\begin{scope}   
\draw [|-|,thick] (-0.5,0) node[left] {} to (0.5,0);
\draw [|-|,thick] (1,0) node[left] {} to (2,0);
\draw [-,thin] (-0.5,0) node[left] {$T_1$} to (2,0);
\draw [|-|,thick] (3,0) node[left] {} to (4,0);
\draw [|-|,thick] (8,0) node[left] {} to (9,0);
\draw [-,thin] (3,0) node[left] {$T_2$} to (9,0);
\end{scope}
\begin{scope}   
\draw [|-|,thick] (4.2,1) node[left] {} to (5,1);
\draw [|-|,thick] (5.3,1) node[left] {} to (6,1);
\draw [|-|,thick] (6.5,1) node[left] {} to (7.5,1);
\draw [-,thin] (4.2,1) node[left] {$T_3$} to (7.5,1);
\end{scope}
\end{tikzpicture}
\caption{A sequential du-opaque history that is not opaque by the definition in \cite{GHS08-permissiveness}}
\label{fig:du-example}
\end{figure*}
%$X\in Rset(T_k)\cap Wset(T_m)$, $T_m$ has committed, such that the response
%of \textit{read}$_k(X)$ precedes the invocation of $\TryC_m()$ in $H$.
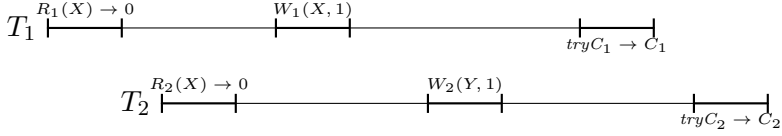
\begin{figure*}[!t]
\begin{tikzpicture}
\node (r1) at (0.5,0) [] {};
\node (w1) at (3.5,0) [] {};
\node (c1) at (7.5,0) [] {};
%\node (r2) at (8.5,0) [] {};

\node (r2) at (2,-1) [] {};
\node (w2) at (5.5,-1) [] {};
\node (c2) at (9,-1) [] {};

\draw (r1) node [above] {\tiny {$R_1(X)\rightarrow 0$}};
\draw (w1) node [above] {\tiny {$W_1(X,1)$}};
\draw (c1) node [below] {\tiny {$\TryC_1 \rightarrow C_1$}};
%\draw (r2) node [above] {\tiny {$R_2(Y) \rightarrow 1$}};

\draw (r2) node [above] {\tiny {$R_2(X)\rightarrow 0$}};
\draw (w2) node [above] {\tiny {$W_2(Y,1)$}};
\draw (c2) node [below] {\tiny {$\TryC_2 \rightarrow C_2$}};
\begin{scope}   
\draw [|-|,thick] (0,0) node[left] {} to (1,0);
\draw [|-|,thick] (3,0) node[left] {} to (4,0);
\draw [|-|,thick] (7,0) node[left] {} to (8,0);
\draw [-,thin] (0,0) node[left] {$T_1$} to (8,0);
\end{scope}
\begin{scope}   
\draw [|-|,thick] (1.5,-1) node[left] {} to (2.5,-1);
\draw [|-|,thick] (5,-1) node[left] {} to (6,-1);
\draw [|-|,thick] (8.5,-1) node[left] {} to (9.5,-1);
\draw [-,thin] (1.5,-1) node[left] {$T_2$} to (9.5,-1);
\end{scope}
\end{tikzpicture}
\caption{History is du-opaque, but not TMS2~\cite{TMS09}}
\label{fig:tms2}
\end{figure*}
\begin{theorem}
\label{th:opeq}
Opacity$_{ut}$=DU-Opacity.
\end{theorem}
\begin{proof}
We show first that every finite history $H\in $ Opacity$_{ut}$ is also du-opaque.
Let $H$ be any finite opaque history such that for any two transactions $T_k,T_m \in
\txns(H)$ that perform $\Write_k(X,v)$ and $\Write_m(X,v)$ respectively, $v\neq v'$
%We claim that $H$ is du-opaque.

Since $H$ is opaque, there exists a final-state serialization $S$ of $H$.
%We claim that each $\Read_k(X)$ in $S$ that does not return $A_k$ is legal in $S_{H}^{r_k(x)}$.
%Suppose by contradiction that $H$ is not du-opaque.
Suppose by contradiction that $H$ is not du-opaque. Thus, there exists a $\Read_k(X)$ that returns a value $v\in V$ in $S$
that is not legal in ${S}_{H}^{k,X}$, the local serialization for $\Read_k(X)$ with respect to $H$ and $S$.
% i.e. $v$ is not latest written value of $X$ in ${S}_{H}^{k,X}$.
Let ${H}^{k,X}$ and ${S}^{k,X}$ denote the prefixes of $H$ and $S$ resp. up to the response of $\Read_k(X)$ in $H$ and $S$ resp..
Recall that the local serialization for $\Read_k(X)$ with respect to $H$ and $S$, ${S}_{H}^{k,X}$ is defined as the subsequence of $S^{k,X}$ that does not contain events of any transaction $T_i \in \txns(H)$ if $H^{k,X}$ does not contain an invocation of $\TryC_i()$.
Since $\Read_k(X)$ is legal in $S$, there exists a committed transaction $T_m \in \txns(H)$ that performs
$\Write_m(X,v)$ that is the latest such write in $S$ that precedes $T_k$.
Thus, if $\Read_k(X)$ is not legal in ${S}_{H}^{k,X}$, the only possibility is that $\Read_k(X) \prec_H^{RT} \TryC_m()$.
Under the assumption of unique writes, there does not exist any other transaction $T_j \in \txns(H)$ that performs $\Write_j(X,v)$.
Consequently, there does not exist any ${\bar H^{k,X}}$ (some completion of $H^{k,X}$) and (t-complete t-sequential history) $S'$ such that $S'$ is equivalent to
${\bar H^{k,X}}$ and $S'$ contains any committed transaction that writes $v$ to $X$ i.e. $H^{k,X}$ is not final-state opaque. However, since $H$ is opaque, every prefix of $H$ must be final-state opaque---contradiction.

By Definition~\ref{def:opaque}, an infinite history $H$ is opaque if
every finite prefix of $H$ is final-state opaque.
Theorem~\ref{th:lc} now implies that
Opacity$_{ut}$ $\subseteq$ DU-Opacity.

By Definition~\ref{def:opaque} and Corollary~\ref{cr:pc}, it follows that DU-Opacity $\subseteq$ Opacity$_{ut}$.
\end{proof}
\subsection{Relation with Other definitions}
\label{sec:def}
Explicitly using the deferred-update semantics in an opacity definition
was first proposed by Guerraoui et al.~\cite{GHS08-permissiveness}
and later adopted by
Kuznetsov and Ravi~\cite{KR11}.
In both papers, opacity is only defined on sequential histories, where
every invocation of a t-operation is immediately
followed by a matching response.
In particular, these definitions require the final-state serialization to respect
the \emph{read-commit order}: $H$ is opaque by their definition if there exists a final-state serialization $S$ of $H$
such that if a t-read of a t-object $X$
by a transaction $T_k$ precedes the tryC of a transaction $T_m$
that commits on $X$ in $H$, then $T_k$ precedes $T_m$ in $S$.
But we observe that this definition is not equivalent to opacity
even for sequential histories.
%: Figure~\ref{fig:uw-example} depicts a counter-example.
In fact the property defined in~\cite{GHS08-permissiveness} is strictly
stronger than du-opacity: the sequential history in Figure~\ref{fig:du-example} is du-opaque (and consequently opaque by Theorem~\ref{th:gkkr}).
%We expect that opaque TM implementations would export such a history.
We can derive a du-opaque serialization $S$ for this history such that
$\ms{seq}(S)= T_1,T_3,T_2$.
In fact, this is the only final-state serialization for $H$. However, by the above definition, $T_2$
must precede $T_3$ in any serialization of this history since the
response of $\Read_2(X)$ precedes the invocation of tryC$_3()$.
Thus, $H$ is not opaque by the definition in
\cite{GHS08-permissiveness}.

The recently introduced \emph{TMS2} correctness
condition~\cite{TMS09, TMS-WTTM} is another
attempt to clarify opacity.
Two transactions are said to \emph{conflict} in a given history if they access the same t-object and at least one of them successfully commits to it.
Informally, for each history $H$ in TMS2, there exists a final-state
serialization $S$ of $H$ such that if two transactions
$T_1$ and $T_2$ conflict on t-object $X$ in $H$, where $X\in \Wset(T_1) \cap \Rset(T_2)$
and tryC of $T_1$ precedes the tryC of $T_2$, then $T_1$ must precede $T_2$ in $S$.
We conjecture that every history in TMS2 is du-opaque,
but not vice-versa.
Figure~\ref{fig:tms2} depicts a history $H$ that is du-opaque, but not TMS2.
Indeed, there exists a du-opaque serialization $S$ of $H$ such that
$\ms{seq}(S)=T_2,T_1$.
On the other hand, $T_1$ and $T_2$ are in conflict, $T_1$ commits before $T_2$, but
there does not exist any final-state serialization of $H$ in which $T_1$ precedes $T_2$.
\section{Discussion}
\label{sec:disc}
It is widely accepted that a correctness condition on a set of histories
should be a safety property, i.e., should be prefix- and limit-closed.
The definition of opacity proposed in~\cite{tm-book}
forcefully achieves prefix-closure by restricting final-state opacity
to prefix-closed histories,
and trivially achieves limit-closure by reducing correctness of an
infinite history to correctness of its prefixes.

This paper proposes a correctness criterion
that explicitly disallows reading from an uncommitted transaction, which
ensures prefix-closure and (under the restriction that every transaction
eventually completes every operation it invokes, but not neccesarily commits or aborts) limit-closure.
We believe that this constructive definition is useful to TM
practitioners, since it streamlines possible implementations of t-read
and tryC operations.
Moreover, it seems that du-opacity already captures the sets of histories
exported by most existing opaque TM implementations~\cite{TL2,DSTM03,noREC}.
In contrast, the recent \emph{pessimistic} STM 
implementation~\cite{pessimisticTM-disc}, 
in which no transaction aborts, 
does not intend to provide the deferred-update
semantics and, thus, is not in the focus of this paper.
Technically, the pessimistic STM of~\cite{pessimisticTM-disc} is not opaque, 
and certainly, not du-opaque.

%However, the notion of deferred-update semantics and consequently du-opacity does not apply to
%models like the pessimistic STM~\cite{pessimisticTM}, in which every transaction eventually commits.
%%The blocking pessimistic TM implementation~\cite{pessimisticTM}, is
%%serializable, and since it has no aborts, is also du-opaque.

To the best of our knowledge,
there is no prior work proving that any TM correctness property
is a safety property in the formal sense.
%There has been recent work on verifying Opacity by specifying it as an automaton and using machine checked simulation proofs
%The main contribution of this paper is to help understand what does
%it mean for Opacity to be safety property.
%It suggests that for most TM implementations, any safe restriction of Opacity must not allow a transaction to read-from a transaction that has not yet invoked \emph{tryCommit}.
The argumentation in the proof of Theorem~\ref{th:lc} is
inspired by the proof sketch  in~\cite{Lyn96} of the safety
of linearizability~\cite{HW90}, but turns out to be trickier due
to the more complicated definition of du-opacity.

\section*{Acknowledgements}

The authors are grateful to the anonymous reviewers
% of WTTM'12 and ICDCS'13
for insightful comments on the previous versions of this paper,
and
% The last author would like to thank
Victor Luchangco for interesting discussions.
% on opacity during PODC'12.

This work was supported by the European Union Seventh Framework Programme
(FP7/2007-2013) under grant agreement Number 238639, ITN project
TRANSFORM, and grant agreement Number 248465, the S(o)OS project.

%
%\newpage
%{\small
\bibliography{references}
%}
%\appendix
%\input{graph}
\end{document}
%%% Local Variables:
%%% mode: latex
%%% mode: flyspell
%%% Local IspellDict: "american"
%%% mode: outline-minor
%%% End: